\documentclass[11pt]{article}

\usepackage{fullpage,amsthm}

\ifx\pdftexversion\undefined
  \usepackage[dvips]{graphicx}
\else
  \usepackage[pdftex]{graphicx}
\fi
\usepackage{amsmath}
\usepackage{yfonts}
\usepackage{smartref} % for referencing the line numbers
\usepackage{color}
\usepackage{paralist} % for compact lists
\usepackage[cmtip,all]{xy} % for long squiggly arrows

\newfont{\mycrnotice}{ptmr8t at 7pt}
\newfont{\myconfname}{ptmri8t at 7pt}

\newtheorem{theorem}{Theorem}
\newtheorem{proposition}[theorem]{Proposition}
\newtheorem{corollary}[theorem]{Corollary}
\newtheorem{lemma}[theorem]{Lemma}
\newenvironment{result-repeat}[2]{\begin{trivlist}
\item[\hspace{\labelsep}{\bf\noindent {#1}~\ref{#2} }]\it}%
{\end{trivlist}}

\newcommand{\journalversion}[1]{}

\newcommand{\gV}{\textswab{V}}
\newcommand{\bR}{\mathbf{R}}

\newcommand{\ceil}[1]{\left\lceil #1 \right\rceil}

\newcommand{\nat}{{\rm I\!N}} 

\newcommand{\mymod}{\mbox{ mod }}

\newcommand{\ignore}[1]{}

\newcounter{linenum}
\addtoreflist{linenum}
\def\codeTabSpace{\hspace*{6mm}}
\newenvironment{code}%
{\begin{tabbing}%
\codeTabSpace \= \hspace*{30mm} \= \hspace*{40mm} \= \hspace*{42mm} \= \kill%
}%
{\end{tabbing}%
}
\newcounter{ind}
\newcommand{\n}{\addtocounter{ind}{7}\hspace*{7mm}}
\newcommand{\p}{\addtocounter{ind}{-7}\hspace*{-7mm}}
\newcommand{\nl}{\\\stepcounter{linenum}{\scriptsize \arabic{linenum}}\>\hspace*{\value{ind}mm}}
\newcommand{\ul}{\\\>\hspace*{\value{ind}mm}}
\newcommand{\bl}{\\[-1.5mm]\>\hspace*{\value{ind}mm}}
\newcommand{\firstline}{\stepcounter{linenum}{\scriptsize \arabic{linenum}}\>}
\newcommand{\lref}[1]{\linenumref{#1}} % use this to refer to a line number
\newcommand{\longsquiggly}{\xymatrix@C=1.5em{{}\ar@{~>}[r]&{}}}
\newcommand{\goes}[1]{\!\overset{#1}{\longsquiggly}\!}

\begin{document}
\sloppy

\title{On the Space Complexity of Set Agreement\thanks{The first two authors were supported by  the Agence Nationale de la Recherche, project
DISPLEXITY, the third author was
    supported by the Agence Nationale de la Recherche,  under
  grant agreement ANR-14-CE35-0010-01, project DISCMAT, and the fourth author
  by the Fondation Sciences
    Math\'e\-ma\-tiques de Paris and the Natural Sciences and
    Engineering Research Council of Canada.}}

\author{
Carole Delporte-Gallet
\protect\footnote{LIAFA, Universit\'e Paris-Diderot} 
\hspace{1cm}
Hugues Fauconnier
\protect\footnotemark[2]
\\
\\
Petr Kuznetsov
\protect\footnote{T\'el\'ecom ParisTech} 
\hspace{1cm}
Eric Ruppert
\protect\footnote{York University}
}

\date{}

\maketitle

\begin{abstract}
The $k$-set agreement problem is a generalization of the classical consensus 
problem in which processes are permitted to output up to $k$ different
input values.
In a system of $n$ processes, an $m$-obstruction-free solution to the problem 
requires termination only in executions where the number of processes taking 
steps is eventually bounded by $m$.
This family of progress conditions generalizes wait-freedom ($m=n$) 
and obstruction-freedom ($m=1$).
In this paper, we prove upper and lower bounds on the 
number of registers required to solve $m$-obstruction-free 
$k$-set agreement, considering both one-shot and repeated formulations.
In particular, we show that repeated $k$ set agreement can be solved using
$n+2m-k$ registers and establish a nearly matching lower bound of $n+m-k$.
\end{abstract}

%\eric{Add category and keywords, etc. later}

\section{Introduction}
\indent

Algorithms that allow processes to reach
agreement are one of the central concerns of the theory of
distributed computing, since some kind of 
agreement underlies many tasks that require
processes to coordinate with one another.
In the classical consensus problem, each process begins with an input
value, and all processes must agree to output one of those input values.
Chaudhuri \cite{Cha93} introduced the $k$-set agreement problem, which
generalizes the consensus problem by
allowing processes to output up to $k$ different input values in any execution.
Consensus is the special case where $k=1$.
Set agreement is trivial for $n$ processes if $k\geq n$:
each process can simply output its own input value.

We consider the $k$-set agreement problem for $k<n$ in an asynchronous 
system equipped with shared read/write registers.
To satisfy {\it wait-free} termination, non-faulty processes
must terminate even if an arbitrary number of processes fail.
The impossibility of solving wait-free $k$-set agreement using registers was a 
landmark result proved by three groups of researchers \cite{BG93b,HS99,SZ00}.
However, Herlihy, Luchangco and Moir~\cite{HLM03} observed that
$k$-set agreement {\it is} solvable (even for $k=1$)
under a weaker termination property, known as {\it obstruction-freedom}
or {\it solo-termination}, which requires that a process must eventually
terminate if it  takes enough steps without
interruption from other processes.
Obstruction-freedom was introduced as a way of separating concerns:
obstruction-free algorithms maintain safety properties in
all possible executions, but make progress only when one process can
run for long enough without encountering contention.
Various scheduling mechanisms designed to reduce contention (such as backing off)
can then be used to satisfy this condition.

Taubenfeld \cite{Tau09a} introduced the $m$-obstruction-freedom
progress property, which requires that, in any execution where
at most $m$ processes take infinitely many steps, each process
that continues to take steps will eventually terminate successfully.
Wait-freedom and obstruction-freedom are special cases,
%of $m$-obstruction-freedom, 
with the extreme values $m=n$ and $m=1$, respectively.
Like ordinary obstruction-freedom, $m$-obstruction-free
algorithms guarantee safety in all runs.  However, $m$-obstruction-freedom
provides a stronger progress property:  larger values 
of $m$ require less rigid constraints on the scheduler
in order to ensure progress.
Since $k$-set agreement has no wait-free solution among $k+1$
processes, it follows that there is no $m$-obstruction free solution when 
$m>k$.
The converse follows from the work of Yang, Neiger and Gafni \cite{YNG98}:
$m$-obstruction-free $k$-set agreement {\it can} be solved 
if $m\leq k$.
%[[PK 
%\eric{Do we have a citation for whoever first observed this?}
%PK: I believe it is the k-set impossibility (HS93 etc.) plus
%k-converge of YNG98
% 
%]]
In this paper, we study how the number of registers required to
solve $m$-obstruction-free $k$-set agreement among $n$ processes
depends on the parameters $m,k$ and $n$.

Previously, the only non-trivial space lower bound was for the
very special case where $m=k=1$.  In this case,
Fich, Herlihy and Shavit~\cite{FHS98} showed $\Omega(\sqrt{n})$
registers are needed.  
The best upper bound for this case is the trivial one of $n$ registers,
which comes from the fact that $n$ (large) single-writer registers
can implement any number of multi-writer registers \cite{VA86}.
Closing the gap between the linear upper bound and the $\Omega(\sqrt{n})$ 
lower bound is a major open problem.  Unfortunately, there has been no progress on this gap in the past two decades.

We first prove nearly tight linear upper and lower bounds
on the number of registers required for {\it repeated} set agreement.
In many applications, such as Herlihy's universal construction \cite{Her91},
there is a sequence of (independent) agreement tasks that 
must be solved, rather than just one.  
We define the repeated $k$-set agreement problem to model this situation.
Processes access an infinite sequence of instances $k$-set agreement in order.  
For all executions and for all $i$, processes accessing the
$i$th instance of $k$-set agreement may output at most $k$ of the values
that are used as inputs to that instance.

We prove that any $m$-obstruction-free solution to 
repeated $k$-set agreement among $n$ processes requires
at least $n+m-k$ registers.
We also give a novel algorithm for this task using
$\min(n+2m-k,n)$ registers.
Previously, the only known set agreement 
algorithm that uses fewer than $n$ registers
was a 1-obstruction-free $k$-set agreement algorithm that uses 
$2n-2k$ registers \cite{DFGR13}.
Our algorithm generalizes that algorithm (to handle any value of $m$) and
improves the number of registers used in the case where $m=1$ from 
$2(n-k)$ to $n-k+2$.
For the case where $m=k=1$, 
our results establish that obstruction-free
repeated consensus requires exactly $n$ registers.
Thus, the gap between the $\Omega(\sqrt{n})$ lower bound and the $O(n)$ upper
bound is closed when we consider the {\it repeated} version of the problem.

For the one-shot version of $k$-set agreement, we focus on the restricted case of anonymous
systems, where processes
do not have unique identifiers and are all programmed identically.
We prove
that any anonymous algorithm must use more than 
$\sqrt{m(\frac{n}{k}-2)}$ registers.
The $\Omega(\sqrt{n})$ lower bound of Fich, Herlihy and Shavit~\cite{FHS98} (for the anonymous case) is a special case of our result with $m=k=1$, but
the new result gives additional insight into the problem by showing
a dependence on $m$ and $k$.  Moreover, the technique used in our proof
is somewhat different, since it requires building an execution
involving many different input values where each process is prevented
from learning about any input value different from its own.
We also prove that it is possible to solve the problem anonymously.
Our algorithm for the repeated version of the problem uses
$(m+1)(n-k)+m^2+1$ registers.
(The usual construction using $n$ single-writer registers
is not applicable, since it presupposes
unique identifiers.)

Figure \ref{summary} summarizes our results.
Our four main results are in boldface; the others are corollaries.
%The lower bounds hold for registers of arbitrary size.

\begin{figure*}
\begin{center}
\begin{tabular}{|c|ll|ll|}\hline
			& \hspace*{25mm}Repeated&& \hspace*{28mm}One-shot &\\\hline
\rule{0pt}{3ex}%
Non-			& {\bf lower}: $n+m-k$ & (Section~\ref{sec-repeated-lower}) 
								& lower: 2 & \cite{DFGR13}\\
\raisebox{-2ex}{\rule{0pt}{5ex}}%
Anonymous	& {\bf upper}: $\min(n+2m-k,n)$ & (Section~\ref{linear-space-alg})
								& upper:	$\min(n+2m-k,n)$ & (Section~\ref{linear-space-alg})\\\hline
\rule{0pt}{3ex}%
Anonymous	& lower: $n+m-k$ & (Section~\ref{sec-repeated-lower}) 
								& {\bf lower}: $\sqrt{m(\frac{n}{k}-2)}$ for $D=\nat$ & (Section~\ref{m-conc-lower-sec}) \\
\raisebox{-2ex}{\rule{0pt}{5ex}}			
			& {\bf upper}: $(m+1)(n-k)+m^2+1$ & (Section~\ref{m-conc-alg-sec})
								& upper: $(m+1)(n-k)+m^2$ & (Section~\ref{m-conc-alg-sec})\\\hline
\end{tabular}
\end{center}
\caption{\label{summary}Lower and upper bounds on the number of registers  to solve $m$-obstruction-free $k$-set agreement among $n$ processes, where $1\leq m\leq k<n$ and input values are from domain $D$ (with $|D|>k$).  Our main results appear in boldface; the others are corollaries.}
\end{figure*}

%\begin{figure}
%\begin{tabular}{|l|l|l|l|}\hline
%One-shot? & Ids? & Lower Bound & Upper Bound\\\hline
%one-shot & anonymous & $\sqrt{m(\frac{n}{k}-2)}$ for large $V$ (Section~\ref{m-conc-lower-sec}) & $(m+1)(n-k)+m^2$ (Section~\ref{m-conc-alg-sec})\\\hline
%repeated & $1..n$ & $n+m-k$\hspace*{23mm}(Section~\ref{sec-repeated-lower}) & 
%$n+2m-k$ if $m<\frac{k}{2}$ (Section~\ref{linear-space-alg})\\
%&&  &
%$n$ \hspace*{17mm}if $m\geq\frac{k}{2}$\\
%\hline
%\end{tabular}
%\caption{\label{summary}Number of registers required to solve $m$-obstruction-free $k$-set agreement among $n$ processes, where $m\leq k<n$ and input values come from the set $V$.}
%\end{figure}

\section{Preliminaries}
\indent

We consider the standard asynchronous shared-memory model, in which $n>1$
processes $p_1,\ldots , p_n$ communicate by applying read and write operations to shared \emph{registers}. The registers
are \emph{multi-writer} and \emph{multi-reader}, i.e., there are no restrictions on 
which processes may access which registers.

%\eric{I commented out the definition of a distributed problem, since it seems unnecessary.}
%A distributed problem specifies the \emph{inputs} the processes
%may receive, the \emph{outputs} they may produce, and a relation between 
%inputs with allowed outputs. 
%In this paper, we focus on the problem of
%(repeated) $k$-set agreement, defined in the next section.
%To solve a problem, every process follows its \emph{algorithm}, 
Each process has a local state that consists of the values stored in its
local variables and a programme counter.
A computation of the system proceeds in {\it steps} performed by the processes.  
Each step is one of the
following: (1)~an invocation of an operation, 
%performs some local computation,
(2)~a read or write operation on a shared register, (3)~local
computation that results in a change of a process's state, or
(4)~a response of an operation.
%\eric{This seems to say that the first shared memory access after an input is part of the same atomic action as the input itself.  Is that too strong?}
%\petr{Do not think it is too strong. I remember this way Sam and
% Vassos did it :)}
Writes update the state of a shared register.
All steps may update the local state of the process that performs it.
A \emph{configuration} %(of an algorithm) 
specifies the state of each register and
the local state of each process at one moment. 
In an \emph{initial configuration},
all registers have the initial values specified by the algorithm
and all processes are in their initial states.

A process is \emph{active} if 
an operation has been invoked on the process but the operation
has not yet produced a matching response; otherwise the process is
called \emph{idle}.  We assume that 
an operation can only be invoked on an idle process and only
active processes take steps.
We focus on deterministic algorithms.
%Each process is modelled as a deterministic state machine. 
Thus, given the current local state of an active process, the 
%transition function of this state machine
algorithm for this process stipulates the unique 
next \emph{step} the process can perform. 
%We treat inputs as invocations of \emph{operations} and outputs as
%the matching responses. 
An \emph{execution fragment} of an algorithm
is a (possibly infinite) sequence of
steps starting from some configuration that ``respects'' 
the algorithm for each process.
An \emph{execution} is an execution fragment that starts from the initial
configuration.
%For any finite execution fragment $\alpha$ and any execution
%fragment $\alpha'$, the execution fragment $\alpha \alpha'$
%is the concatenation of $\alpha$ and $\alpha'$;
%in this case $\alpha'$ is called an \emph{extension} of $\alpha$.
% I deleted next line, since it would also be true of nondeterministic algs--Eric
%Since we only consider deterministic algorithms, every finite execution
%unambiguously identifies the resulting configuration. 
%We assume that inputs and outputs alternate at every given process.
An operation is completed if its invocation is followed by
a matching response.
In an infinite execution, a process is \emph{correct} if it takes an infinite number of steps
or is idle from some point on. %; otherwise, it is \emph{faulty}.

%[[PK should go to the upper bounds
%Without loss of generality, we also assume that processes may take \emph{atomic 
Our algorithms make use of multi-writer \emph{snapshot objects}~\cite{AAD93},
which can be implemented from registers.
A snapshot object stores a vector of $r$ values and 
provide two atomic operations:
$\textit{update}(i,v)$ ($i\in\{1,\ldots,r\}$), which writes value $v$ to
component $i$, and $\textit{scan}()$, 
which returns the vector of the most recently written values to
components $1,\ldots,r$.

%\eric{I deleted the definition of obstruction-free executions, because I don't think it
%makes intuitive sense to talk about executions being free of obstructions.  It is
%{\it algorithms} that are free of obstructions.}
%An infinite execution $\alpha$ is \emph{$m$-obstruction-free} (for
%$m\geq 1$) if at
%most $m$ processes take infinitely many steps in $\alpha$. 
%A $1$-obstruction-free execution in which $p$ is the process that
%takes infinitely many steps is also called \emph{$p$-solo}. 

%An execution fragment $\alpha$ is \emph{$Q$-free},
%for a non-empty set of processes $Q$,
%if $\alpha$ contains no steps by  processes in $Q$.
%(If $Q=\{q\}$, we simply write that $\alpha$ is $q$\emph{-free}.)

%We denote by $\alpha \vert p$  the subsequence of events of
%execution $\alpha$ that are applied by process $p$.
%Two executions are \emph{indistinguishable} to a process $p$,
%if $p$ applies exactly the same sequence of events and gets the
%same responses from these events in both executions.

\medskip

\subsection{Set agreement}
We begin with a formal definition of the \emph{repeated $k$-set agreement} problem.
Processes may perform  {\sc Propose}($v$) operations, where $v$ is drawn from an input domain $D$. 
Each {\sc Propose} operation outputs a response from $D$ when it terminates. 
For an execution $\alpha$, let $In_i(\alpha)$ be the set of values that are used as the
argument to some process's $i$th invocation of {\sc Propose} and let $Out_i(\alpha)$ be the
set of values that are the response of some process's $i$th  {\sc Propose} operation.  
Then,  in every execution $\alpha$ of an algorithm that solves repeated
$k$-set agreement the following properties must hold.
% that satisifies the properties of the system
%($t$-resilient, $m$-obstruction-free or whatever).
\begin{itemize}
\item
\emph{Validity}:  $\forall i$, $Out_i(\alpha)\subseteq In_i(\alpha)$.
\item
\emph{$k$-Agreement}:  $\forall i$, $|Out_i(\alpha)|\leq k$.
\end{itemize}
An $m$-obstruction-free algorithm must additionally satisfy the following termination condition.
\begin{itemize}
\item
\emph{$m$-Obstruction-Freedom}: 
in every execution in which 
%[[PK simpler?
%the number of processes that take infinitely many steps is
at most $m$ processes take infinitely many steps, 
%]]
every correct process
completes each of its  operations. 
%there is no suffix of an execution in which at most $m$ processes take infinitely many steps and no {\sc Propose} operations terminate.
\end{itemize}
%[[PK ok
%\eric{I modified the definition of m-o-f termination.  The old version said
%if at most m processes are correct then every correct process finishes
%each of its proposes, which is weaker than what I wrote above.  For example,
%if all processes are correct but only m processes take infinitely many steps
%(the others stay idle), then we also want to guarantee termination.}
%]]

The special case when $k=1$ is called \emph{consensus}.
In the \emph{(one-shot) $k$-set agreement problem}, every process
invokes {\sc Propose} at most once.

It is known that wait-free $(k+1)$-process $k$-set agreement cannot
be solved using registers~\cite{BG93b,HS99,SZ00}.  This implies the following lemma,
which we shall use to prove our space lower bounds.

\begin{lemma}\label{lem:m-val}
Let $A$ be any algorithm that solves $m$-obstruction-free $k$-set
agreement using registers.  
For any set $V$ of $m$ input values and any set $Q$ of $m$ processes, there is an execution of $A$ in
which only processes in $Q$ take steps and all values in $V$ are output.
\end{lemma}
\begin{proof}
Suppose the opposite for some sets $V$ and $Q$  and consider all executions of $A$ in which only 
processes in $Q$ with inputs in $V$ take steps. By the assumption, at most $m-1$ distinct
values are decided in each of these executions, which implies a wait-free
$m$-process $(m-1)$-set agreement algorithm, violating~\cite{BG93b,HS99,SZ00}.   
\end{proof}

Lemma~\ref{lem:m-val} implies that no algorithm
can solve  $m$-obstruction-free $k$-set agreement using registers if $k<m$.  
In the rest of the paper, we derive lower and upper bounds on the
space complexity of $m$-obstruction-free %(repeated)
$k$-set agreement for $n$ processes, where $m\leq k < n$.  
(If $k\geq n$, the problem is trivial and no registers are required: each process can simply output its own input value.)

% !TEX root =  podc-tr.tex

\section{Lower Bound for Repeated Set Agreement}
\label{sec-repeated-lower}
\indent

In this section, we prove that solving $m$-obstruction-free repeated $k$-set agreement among $n$ processes requires at least $n+m-k$ registers.  Since the proof is technical, we first provide
a brief overview.  For simplicity, assume for now 
that $k+1$ is a multiple of $m$.
We assume that there is an algorithm
that uses fewer than $n+m-k$ registers, and construct
an execution in which processes return $k+1$ different values in
some instance of set agreement, contradicting the $k$-agreement property.
The proof first constructs $c=\frac{k+1}{m}$ disjoint sets $Q_1,Q_2, \ldots, Q_{c}$ of $m$ processes each, and an execution $\alpha$ 
that passes through a sequence
of configurations $D_1,D_2, \ldots, D_{c}$ with the following property. 
For $1\leq i<c$, 
{\it every possible} execution fragment by the processes in $Q_i$ starting from $D_i$ 
writes only to registers that are overwritten immediately after $D_i$ 
in $\alpha$.
Moreover, processes in $Q_i$ take no more steps after $D_i$ in $\alpha$.
We can then splice into $\alpha$ 
any execution fragment by processes in $Q_i$ at $D_i$, knowing that 
the rest of $\alpha$ will not be affected, since all evidence
of the inserted steps will be overwritten.
For each group $Q_i$, the fragment we splice into $\alpha$
accesses
a ``fresh'' instance of set agreement that was never accessed during
$\alpha$.  (In each fragment that is spliced in, only the $m$
processes in $Q_i$ take steps, so all {\sc Propose} operations terminate
and the processes will eventually reach and complete 
the fresh instance of set agreement.)
We ensure that these groups of $m$ processes output
disjoint sets of $m$ different values each for this one instance of set agreement, for a total of  $c\cdot m = k+1$ different outputs, a contradiction.

\begin{theorem}
\label{repeated-lower-bound}
Any algorithm for $m$-obstruction-free repeated $k$-set agreement among $n$ processes requires at least $n+m-k$ registers.
\end{theorem}

\begin{proof}
To derive a contradiction, assume there exists an algorithm for $m$-obstruction-free repeated $k$-set agreement using $r=n+m-k-1$ registers.
Let $c=\ceil{\frac{k+1}{m}}$.  Since $k\geq m$, we have $c\geq 2$.
We define a {\it block write} to a set $A$ of registers by a set $P$ of processes
to be an execution fragment in which each process of $P$ takes a single step,
such that the set of registers written during the fragment is $A$.

We first construct
an execution
\begin{equation}
\label{execution}
C_0 \goes{\alpha_1} D_1 \goes{\beta_1} C_1 \goes{\alpha_2} D_2 \goes{\beta_2} C_2 \goes{\alpha_3} \cdots \goes{\beta_{c-1}} C_{c-1}
\end{equation}
and sets $A_1,\ldots,A_{c-1}$
of registers such that $C_0$ is the initial configuration and for all $j$,
\begin{enumerate}
\item
\label{alphaCD}
$\alpha_j$ is an execution fragment containing only steps by two disjoint sets $P_j$ and $Q_j$ of processes that goes from configuration $C_{j-1}$ to configuration $D_j$,
\item
\label{betaDC}
$\beta_j$ is a block write to $A_j$ by $P_j$ that goes from configuration $D_j$ to configuration $C_j$,
%\item
%\label{inside}
%no process writes outside $A_j$ during $\alpha_j$,
\item
\label{Q1size}
$|Q_1| = k+1-(c-1)m$,
\item
\label{Qjsize}
if $j>1$, $|Q_j|=m$,
\item
\label{PAsize}
$|P_j|=|A_j|$,
\item
\label{Qdisjoint}
$Q_j \cap Q_{j'} =\emptyset$ for $j'\neq j$,
\item
\label{PQdisjoint}
$Q_j\cap P_{j'} =\emptyset$ for $j'>j$, and
\item
\label{no-continuation}
there is 
no execution fragment starting from $D_j$ in which only processes in $Q_j$ take steps and some process writes outside $A_j$.
\end{enumerate}
%[[PK ok
%\eric{Note:  I deleted one of the properties (which said no process writes outside $A_j$ during $\alpha_j$), since I didn't see why it was needed.}
%We describe how to construct the prefix of this execution up to $C_j$ by induction on $j$.
%]]

{\sc Base case} ($j=0$):  Let $C_0$ be the initial configuration.

\smallskip

{\sc Inductive step}:  Let $1\leq j\leq c-1$.  Assume we have constructed the execution from $C_0$ to $C_{j-1}$ satisfying all the properties.
The algorithm in Figure \ref{construction-alg} constructs the execution fragment $\alpha_j$
%$C_{j-1} \goes{\alpha_j} D_j$ 
and the sets $P_j$, $Q_j$ and $A_j$.
Then, let $\beta_j$  be the execution fragment starting from $D_j$ where each process in $P_j$ takes a single step and let $C_j$ be the resulting configuration.

\begin{figure*}
\setcounter{linenum}{0}
\begin{code}
\firstline
let $\alpha_j$ be the empty execution fragment\nl
$D_j \leftarrow C_{j-1}$\nl
$P_j\leftarrow \emptyset$\nl
$A_{j}\leftarrow \emptyset$\nl
if $j>1$ then $size \leftarrow m$ else $size \leftarrow k+1-(c-1)m$\nl
let $Q_j$ be a set of $size$ processes disjoint from $Q_1 \cup Q_2 \cup \cdots \cup Q_{j-1}$
\label{choose-proc1}\nl
loop until no execution fragment starting from $D_j$ by $Q_j$ writes outside $A_j$\nl
\n   let $\delta$ be an execution fragment starting from $D_j$ by $Q_j$ until some process $q\in Q_j$ is poised for\ul
\>		  the first time to write to a register that is not in $A_j$ and let $R$ be that register\nl
     $\alpha_j \leftarrow \alpha_j \cdot \delta$\nl
     let $D_j$ be the configuration reached from $C_{j-1}$ by performing $\alpha_j$\label{updateDj}\nl
     let $q'$ be some process outside $Q_1\cup Q_2\cup \cdots \cup Q_j \cup P_j$\label{choose-proc2}\nl
	 $A_j \leftarrow A_j \cup \{R\}$\label{updateAj}\nl
     $P_j \leftarrow P_j \cup \{q\}$\label{updatePj}\nl
     $Q_j \leftarrow (Q_j - \{q\}) \cup \{q'\}$\label{updateQj}\nl
\p end loop\nl
output $\alpha_j,D_j,P_j,Q_j,A_j$
\end{code}
\caption{Algorithm used in %the induction step of 
the proof of Theorem \ref{repeated-lower-bound} to construct $\alpha_j, D_j, P_j, Q_j$ and $A_j$.\label{construction-alg}}
\end{figure*}

Observe that the construction algorithm terminates:
each loop iteration adds a new register to $A_j$,
so it terminates after at most $r$ iterations.
We next check that the required processes on line
\lref{choose-proc1} and \lref{choose-proc2} exist.
When $j=1$, we have $size=k+1-(c-1)m =k+1-\ceil{\frac{k+1}{m}}\cdot m +m\leq m < n$, so one can choose the required processes on line~\lref{choose-proc1}.
For $j>1$, one can choose the process on line~\lref{choose-proc1} because 
\begin{eqnarray*}
|Q_1\cup\cdots\cup Q_{j-1}| 
&=& k+1-(c-1)m + (j-2)m \\
&&\mbox{(by induction hypothesis \ref{Q1size}, \ref{Qjsize} and \ref{Qdisjoint})}\\
&\leq& k+1 - (c-1)m + (c-3)m\\
&& (\mbox{since }j\leq c-1)\\ 
&=& k+1-2m 
\leq  n-2m\\
&&(\mbox{since }k<n).
\end{eqnarray*}
Similarly, one can choose the required process $q'$ at line \lref{choose-proc2} because
\vspace*{-1mm}
\begin{eqnarray*}
&&|Q_1\cup \cdots\cup Q_j\cup P_j|\\
&\leq& k+1-2m +|Q_j| + |P_j|  \\
&&(\mbox{since } |Q_1\cup\cdots\cup Q_{j-1}|\leq k+1-2m)\\
&\leq & k+1-m + r - 1 \\
&&(\mbox{since $|Q_j|=m$ and }|P_j|=|A_j| \leq r-1)\\
&=& n -1 \\
&&(\mbox{since }r=n+m-k-1).
\end{eqnarray*}

We verify the construction satisfies all of the properties.
Line \lref{updateDj} of the algorithm updates $D_j$ each time $\alpha_j$ is updated, to ensure  property 1.
Property \ref{betaDC} is true by definition of $\beta_j$ and $C_j$.
$Q_j$ is initialized to a set whose size satisfies property \ref{Q1size} or \ref{Qjsize} on line \lref{choose-proc1} and the size of this set is preserved whenever $Q_j$ is altered on line \lref{updateQj}.
$P_j$ and $A_j$ are initialized to be empty, and both are updated by adding one element to each on line \lref{updateAj} and \lref{updatePj}, so they remain the same size after every iteration of the loop.  (Note that $P_j$ and $Q_j$ are disjoint at the beginning of each iteration of the loop, so line \lref{updatePj} does add a new process to $P_j$.)
Every process placed in $Q_j$ at line \lref{choose-proc1} or \lref{updateQj} was chosen to be outside $Q_1\cup \ldots\cup Q_{j-1}$, guaranteeing property
\ref{Qdisjoint}.
Similarly, processes added to $P_j$ are always outside $Q_1\cup\ldots\cup Q_{j-1}$, and whenever a process is added to $P_j$, it is removed from $Q_j$, so property \ref{PQdisjoint} is satisfied.
Finally, property \ref{no-continuation} is guaranteed by the exit condition of the loop.
This completes the inductive construction.

\smallskip

Now, let $s$ be the maximum number of invocations of {\sc Propose} by any process in the execution that takes the system to configuration $C_{c-1}$.
Let $Q_c$ be a set of $m$ processes disjoint from $Q_1\cup\cdots \cup Q_{c-1}$.
(These $m$ processes exist since $|Q_1\cup \cdots\cup Q_{c-1}| = k+1-m \leq n-m$.)
Let $D_c=C_{c-1}$.

For each $j\in\{1,\ldots,c\}$, we now construct an execution fragment $\gamma_j$
by the processes in $Q_j$ starting from $D_j$.
Since $|Q_j|\leq m$, each {\sc Propose} in $\gamma_j$ must terminate.
First, the processes in $Q_j$ run one by one until each
completes its first $s$ invocations of {\sc Propose}.
Then, the processes of $Q_j$ run their $(s+1)$th invocation of {\sc Propose},
each using its own id as its input value so that they decide $|Q_j|$ different output values.  
By Lemma~\ref{lem:m-val}, such an execution fragment exists.
%; otherwise we could solve
%$|Q_j|$-obstruction-free $(|Q_j|-1)$-set agreement using only
%registers, which is impossible (Lemma~\ref{lem:}).
%\eric{Or should this just cite \ref{thm:kset-mof}; then there is no need for corollary}
Note that for $j<c$, $\gamma_j$ cannot write outside of $A_j$, by property \ref{no-continuation}.  So, all traces of $\gamma_j$'s activity are obliterated by the block
write $\beta_j$.  Thus, we can insert $\gamma_1,\ldots,\gamma_{c}$ into
execution (\ref{execution}) at $D_1,\ldots,D_c$, respectively, 
and the resulting execution
is still legal.  In the resulting execution, the number of distinct outputs
for the $(s+1)$th instance of set agreement is $\sum\limits_{j=1}^c|Q_j| = k+1$, violating $k$-agreement.
This completes the proof.
\end{proof}

\section{Algorithm for Repeated Set Agreement}
\label{linear-space-alg}

\subsection{One-shot $k$-set agreement}
%For ease of exposition, 
We first give an 
algorithm that uses a snapshot object of $r=n+2m-k$ components to solve 
(one-shot) $m$-obstruction-free $k$-set agreement, and then describe
how to extend it to solve repeated set agreement.
The one-shot algorithm is shown in Figure~\ref{two-copies}.
Roughly speaking, the first $k-m$ processes to decide
can output arbitrary values, but we ensure that the last $\ell=n-k+m$ processes
all agree on at most $m$ different values (for a total of at most $k$ different values).

Each process stores its preferred value in its local variable {\it pref}.  Initially, it prefers
its own input value.  
Each process executes a loop in which it stores its {\it pref} and identifier into
a component of the snapshot object, takes a scan of the snapshot object and updates
its {\it pref} variable based on the scan.
The location $i$ that the process updates  advances in each iteration of the loop, as long
as the process's {\it pref} value remains the same.
When the process updates its {\it pref}, it does not advance to the next location:  instead
it updates  the same location during the next iteration of the loop.

The process repeats this loop until a scan returns a vector containing at most $m$ different 
value-id pairs, at which point it returns one of those values.
In each iteration, a process updates its {\it pref} value when it does not see any copies of its current
value-id pair anywhere in the vector returned by the scan, except for the component
it just updated, {\it and} it does see two copies of some other pair.
In this case, it adopts the value of the pair that appears twice as its {\it pref}.

The algorithm in Figure \ref{two-copies} is an improvement on the algorithm
of \cite{DFGR13}, which was designed for the special case where $m=1$ and uses $2(n-k)$
registers, compared to the $n-k+2$ registers used by ours.

\setcounter{linenum}{0}
\begin{figure*}
\begin{code}
\firstline
Shared variable:\nl
\n $A$: snapshot object with $r=n+2m-k$ components, each initially $\bot$\bl\nl
\p {\sc Propose}($v$)\nl
\n  {\it pref} $\leftarrow v$\nl
    $i\leftarrow 0$\nl
    loop \label{beginloop}\nl
\n      update $i$th component of $A$ with $(\mbox{\it pref},id)$\label{write-pref}\nl
        $s\leftarrow $ scan of $A$ \label{snap}\nl
        if $|\{s[j] : 0\leq j < r\}| \leq m$  and $\forall j$, $s[j]\neq\bot$  then \nl
        \n let $j_1 \leftarrow \min\{j_1 : \exists j_2>j_1 \mbox{ such
          that }s[j_1]=s[j_2]\}$, output value in $s[j_1]$ and halt\label{output}\nl

\p if $\forall j\neq i, s[j]\notin\{\bot, (\mbox{\it pref},id)\}$ and
$\exists j_1\neq j_2$ such that $s[j_1]=s[j_2]$ then \label{cond}\nl 
 \n	$j_1\leftarrow \min\{j_1 : \exists j_2>j_1 \mbox{ such that }s[j_1]=s[j_2]\}$\nl
            {\it pref} $\leftarrow$ value in $s[j_1]$ \label{change-pref}\nl
\p      else $i\leftarrow (i+1) \mymod r \label{change-ind}$\nl
\p  end loop\nl
\p end {\sc Propose}
\end{code}
\caption{Algorithm for $m$-obstruction-free $k$-set agreement.  Code for a process with identifier $id$.\label{two-copies}}
\end{figure*}

We now prove that the algorithm in Figure~\ref{two-copies} indeed
solves $m$-obstruction-free $k$-set agreement. 
It is easy to see that {\bf validity} holds: the only values that can appear in the snapshot object or in a process's local {\it pref} variable are input values.  Thus, only input values can
be produced as outputs.  

Before proving $k$-agreement and termination, we 
first establish the following invariant. %(in fact always true property):
\begin{lemma}\label{invOne}
For each process identifier $id$, all the pairs in $A$ with identifier $id$ have the same value. 
\end{lemma}

\begin{proof}
To derive a contradiction, assume there is an execution that reaches 
a configuration $C$ in which
$A[i_1]=(v_1,id)$ and $A[i_2]=(v_2,id)$ where $v_1\neq v_2$.
%Consider an execution in which $C$ is reached at some time~$\mu$.
Let $p_{id}$ be the process with identifier $id$.
Let $u_1$ and $u_2$ be the last steps before $C$ in which $p_{id}$ updates 
$A[i_1]$ and $A[i_2]$, respectively.
Without loss of generality, assume $u_1$ is before $u_2$.
Then, between $u_1$ and $u_2$, $p_{id}$ changes its {\it pref}
variable at line~\lref{change-pref}.  Consider the first time after
$u_1$ that $p_{id}$ performs such a change, and let
$i^*$ and $s^*$ be the values of $p_{id}$'s local variables $i$ and $s$ at that time.
Since $s^*$ was obtained from a scan between $u_1$ and $C$ and $A[i_1]=(v_1,id)$ throughout that interval, $s^*[i_1]$ is $(v_1,id)$.
Thus, $i^*=i_1$; otherwise the test on line \lref{cond} would not be satisfied,
and $p_{id}$ would not change {\it pref} at line~\lref{change-pref}.
Therefore, in the next iteration of the loop, $p_{id}$ will update location $A[i_1]$.
This update is after $u_1$ and no later than $u_2$ (and hence before
$C$), which contradicts the definition of $u_1$ as the last update
performed by $p_{id}$ on $A[i_1]$ before~$C$.
\end{proof}

To prove {\bf $k$-agreement}, 
%the proof is similar to the proof for the algorithm in
%Figure \ref{majority}.  \eric{Can we somehow unify the two proofs?}
let $\ell=n-k+m$. 
If at most $n-\ell$ processes decide, then $k$-agreement is trivial since
$n-\ell = k-m < k$.
So, consider an execution in which more than $n-\ell$ processes decide.
Order the processes that decide according to the times when each performs its last scan,
and let $q_0$ be the $(n-\ell+1)$th process in this ordering.
Let $X$ be the set of at most $m$ different pairs that appear in the vector that
$q_0$'s final scan returns. Let $V$ be the set of values that appear in pairs of $X$. 
Then, $|V|\leq|X|\leq m$.
We prove that $q_0$ and all processes that come later in the ordering
output values in~$V$.
Thus, the total number of values output is at most $(n-\ell) + |V| \leq n-(n-k+m)+m = k$.

\begin{lemma}\label{lem:safety}
In any configuration after $q_0$ performs its final scan, 
only pairs with values in $V$ can appear in two or more locations of $A$.
%no process writes a pair with a value not in $X$ into the snapshot 
%object more than once.
\end{lemma}
\begin{proof}
%\eric{Should be possible to do this in the same was as for algorithm in Figure \ref{majority}.}
Let $C_0$ be the configuration just after $q_0$'s final scan.
We shall show by induction that in each configuration
reachable from $C_0$, 
%no pair with a value not in $V$ 
only pairs with values in $V$ can appear in two or more locations of~$A$.
For the base case, consider the configuration $C_0$.
By the definition of $V$, $A$ contains only pairs with values in $V$,  
so the claim holds.

For the induction step, suppose the claim holds in all configurations 
from $C_0$ to some configuration $C_1$ reachable from $C_0$.
Let $st$ be a step that takes the system from $C_1$ to another configuration $C_2$.
We show that the claim holds in configuration $C_2$.
We need only consider the case where $st$ is an update by some process 
$p_{id}$.  Let $(v,id)$ be the pair that $st$ stores in a component of $A$.
%If $st$ writes a pair with a value in $V$, it preserves the claim.
%So, for the remainder of the proof we consider the case where 
%a process $p_{id}$ writes a pair $(v,id)$ into $A$, where $v\notin V$.

{\sc Case 1}: $st$ is the first update by $p_{id}$ after $C_0$. 
If $v\in V$, then $st$ cannot cause a violation of the claim.
If $v\notin V$, then $A$ contains exactly one copy of $(v,id)$ in configuration
$C_2$, since $(v,id)\notin X$, so again $st$ preserves the claim.

{\sc Case 2}: $st$ is not the first update by $p_{id}$ after $C_0$.
Let $s_{id}$ be the vector obtained by $p_{id}$'s
last scan before $st$.
We show that $v\in V$, and hence $st$ preserves the claim, by considering two subcases.

{\sc Case 2a}:  $s_{id}$ satisfies the condition on line \lref{cond}. 
Then, $p_{id}$ updates its {\it pref} variable at line \lref{change-pref},
so the value $v$ is the value
of a pair that appears twice in $s_{id}$.
By the induction hypothesis, $v\in V$.

{\sc Case 2b}: $s_{id}$ does not satisfy the condition on line \lref{cond}.
We first argue that at least one pair appears twice in $s_{id}$.
Recall that there are at most $\ell-1$ undecided processes in $C_0$.
Since $A$ contains at most $m$ distinct pairs ($|X|\leq m$) in $C_0$
and at most $\ell-1$ processes update $A$ after $C_0$,
Lemma~\ref{invOne} implies that, 
%in each configuration reachable from $C_0$, 
when the scan $s_{id}$ is performed, $A$ contains at most
$m+\ell-1=n+2m-k-1$ distinct pairs. Since there are $r=n+2m-k$ 
locations in $A$, at least one pair appears twice in $s_{id}$. 

Since $q_0$ has previously output a value, $s_{id}$ contains no $\bot$ elements. 
Thus, the reason that $s_{id}$ does not satisfy the condition on line \lref{cond}
must be that for some $j$ different from $p_{id}$'s position $i$,
$s_{id}[j]=(\textit{pref},id)$.  Just before taking the scan $s_{id}$, $p_{id}$ stores
$(v,id)$ in location $i$. This update occurs after $C_0$, since $st$ is 
not the first update by $p_{id}$ after $C_0$.  
In the configuration after this update of location $i$,
both $s_{id}[j]$ and $s_{id}[i]$ contain $(v,id)$.
So, by the induction hypothesis, $v\in V$.
\end{proof}
Lemma~\ref{lem:safety} implies that all processes after the $(n-\ell)$th in the ordering can only decide one of the (at most) $m$ values in $V$ and, thus, {\bf $k$-agreement} is ensured. 

To prove {\bf $m$-obstruction-freedom}, consider an execution where
the set $P$ of processes that take infinitely many steps has size at most $m$.
To derive a  contradiction, assume some process in $P$ never decides.
In each  loop iteration, a process either
keeps its preferred value and increments $i$ (its location to
update) modulo $r$ or sets its preferred value without modifying  $i$. 
We partition $P$ into two subsets: the set $NS$
of ``non-stabilizing'' processes that modify  $i$ 
infinitely often and
the set $S$ of  ``stabilizing'' processes that eventually get stuck 
updating the same location $i$ forever.
% and forever changing their preferred
%values each time they execute the loop.
\journalversion{Every process in $NS$ updates each component of $A$
infinitely often, and there is a time after which each process in
$S$, whenever it executes the loop, changes its preferred
value and stores it in the same location.
}

\begin{lemma}\label{term-dontchange}
There is at least one process in $NS$.
\end{lemma}
\begin{proof}
To derive a contradiction, assume the claim is false (i.e., $P=S$).
Let $\mu$ be a time after which only processes in $P$ take steps and
no process changes its local variable $i$.
Then there is a set $M$ of at most $m$ locations whose contents are updated after $\mu$.  Let $NM$ be the set of at least $n+m-k\ge 2$ locations that are not
updated after $\mu$.
%Let $NM$ be the set of registers that are not modified ($|NM|\ge 2$), and $M$ ($|M| \le m$) be the set of registers that are modified.
Let $\mu'$ be any time when each process in $P$ has performed at least one
update after $\mu$. Thus, at $\mu'$, every location in $M$ contains a pair stored by a process in $P$.

Let $p$ be a process in $P$ that performs a scan that returns a vector $s_p$ after $\mu'$.  
By the hypothesis, $p$ changes its preferred value in every iteration after $\mu'$, so $s_p$ satisfies the condition on line \lref{cond}.
Process $p$ then changes {\it pref} to a value $v$ in a pair $(v,k)$ 
that appears twice in $s_p$.
Since each component in $M$ is updated by different processes, no two can
contain the same pair after $\mu'$.
We consider two cases.

{\sc Case 1:}  in $s_p$, $(v,k)$ appears in one component of $M$ and one of $NM$. 
As  $(v,k)$ is read from a component in $M$ after $\mu'$, $p_k\in P$.
Consider the time (after $\mu$) at which $p_k$ stores $(v,k)$ in a
component in $M$.
Since no register in $NM$ ever changes its value after $\mu$, 
in $p_k$'s subsequent scan, %and by Lemma~\ref{invOne}, 
$(v,k)$ is in some register of $NM$ and $p_k$ will not change its preferred value, contradicting the fact that $P=S$.

{\sc Case 2:} in $s_p$, $(v,k)$ appears in two components of $NM$.
By the definitions of $NM$ and $\mu'$,  $(v,k)$ is found twice in $NM$ at all times after $\mu$. 
As $p$ changes its preferred value after its next update, it must have found another pair that appears twice and was not in $A$ previously. 
Then this new pair cannot be in two locations in $NM$. The pair cannot  be in two locations
in $M$ either because all the locations of $M$ are updated by different processes. 
Thus, this new pair is  in one location of $M$ and one location of $NM$. But, as we have seen in
Case 1, this leads to
a contradiction.
\end{proof}
%
%By the definition of $NS$ and Lemma~\ref{term-dontchange}, at least one
%process in $P$ updates each component of $A$ infinitely often. Hence
%we have the following corollary.
Thus, some process updates each component of $A$ infinitely often, yielding the following corollary.
\begin{corollary}% (of Lemma~\ref{term-dontchange}) 
\label{term-cor}
There is a time  after which $A$ contains only pairs stored by processes in $P$.
\end{corollary}

By Corollary~\ref{term-cor}, there is a time $\nu$ after which
(1)~$A$ contains only pairs stored by processes in $P$. By
Lemma~\ref{invOne}, (2)~all pairs in $A$ with the same id have the same
value. By the assumption, (3)~$|P|\leq m$. (1), (2) and (3) imply that
after $\nu$, each time a  process  $p\in P$ performs a scan 
it finds at most  $m$ different pairs in the snapshot and 
decides.  This contradiction establishes the {\bf $m$-obstruction-freedom} property.

\begin{theorem}
\label{one-shot-alg}
For $1\leq m\leq k<n$, there is an $m$-obstruction-free algorithm that solves $k$-set agreement among $n$ processes
using $\min(n+2m-k,n)$ registers.
\end{theorem}

\begin{proof}
We established above that the algorithm in Figure~\ref{two-copies} solves the problem using
a snapshot object of $n+2m-k$ components.
If $n+2m-k\leq n$, the snapshot object can be implemented from $n+2m-k$ registers~\cite{EFR07}. 
Otherwise, the snapshot  can be implemented from $n$ single-writer registers \cite{AAD93,VA86}.
%\petr{Not sure we need to talk about ``the set of ids known in
% advance here.}
\journalversion{If the set of process ids is known a priori, then this completes the
proof.
Otherwise, the $n$ single-writer registers can be implemented from $n$ multi-writer registers
in a non-blocking manner \cite{DFGR15}.}
\end{proof}

\medskip

\subsection{Repeated $k$-set agreement}
The one-shot $k$-set agreement algorithm can be transformed into an algorithm for
repeated set agreement with the same space complexity to prove the following theorem. 
Since it is quite similar to the one-shot algorithm, we
describe it briefly.  

\setcounter{linenum}{0}
\begin{figure*}
\begin{code}
\firstline
Shared variable:\nl
\n $A$: snapshot object with $r=n+2m-k$ components, each initially $\bot$\bl\nl
\p 
Persistent local variables:\nl
 \n   $i\leftarrow 0$\nl
 $t \leftarrow 0$\nl
 $\textit{history} \leftarrow $ empty sequence\bl\nl    
\p {\sc Propose}($v$)\nl
\n $t \leftarrow t+1$\nl 
     if $|\textit{history}|\geq t$ then   \nl

   \n output the $t$-th value in \textit{history} and  halt\label{output-old-rep}\nl
 \p   {\it pref} $\leftarrow v$\nl 

    loop \label{beginloop-rep}\nl
\n      update $i$th component of $A$ with $(\mbox{\it pref},id,t,\textit{history})$\label{write-pref-rep}\nl
        $s\leftarrow $ scan of $A$ \label{snap-rep}\nl
        if $\exists j$ such that $s[j]=(w,id',t',his)$ with $t' > t$ then \label{halt-cond1-rep} \nl
        \n  $\textit{history} \leftarrow his$, output the
               $t$-th value in $his$ and halt \label{halt-rep}\nl 
\p      if $|\{s[j] : 0\leq j < r\}| \leq m$  and 
$\forall j$, $s[j]$ is neither
        $\bot$ nor of the form $(w,q,t',his)$ with $t'<t$
then \label{halt-cond2-rep}\nl
        \n  let $j_1 \leftarrow \min\{j_1 : \exists j_2>j_1 \mbox{ such
                that }s[j_1]=s[j_2]\}$\nl
            let $w$ be value in $s[j_1]$\nl
            $\textit{history} \leftarrow \textit{history}\cdot w$\nl 
            output $w$ and halt\label{output-rep}\nl
\p      if $\forall j\neq i, s[j]\notin\{\bot, (\mbox{\it pref},id,t,\textit{history})\}$ and $\exists j_1\neq j_2$ such that $s[j_1]$ and $s[j_2]$ contain \ul
\n\n               identical $t$-tuples then \label{cond-rep}\nl 
\p	        $j_1\leftarrow \min\{j_1 : \exists j_2>j_1 \mbox{ such that $s[j_1]$ and $s[j_2]$ contain identical $t$-tuples}\}$\nl
            {\it pref} $\leftarrow$ value in $s[j_1]$ \label{change-pref-rep}\nl
\p      else $i\leftarrow (i+1) \mymod r \label{change-ind-rep}$\nl
\p  end loop\nl
\p end {\sc Propose}
\end{code}
\caption{Algorithm for $m$-obstruction-free repeated $k$-set agreement. \label{mconc-rep-kset}}
\end{figure*}

The pseudocode for our repeated $k$-set agreement algorithm is given in
Figure~\ref{mconc-rep-kset}.  
It essentially follows the pseudocode of the  one-shot algorithm
(Figure~\ref{two-copies}), with additional ``shortcuts'' which a
process may use to adopt a value output previously by another process that has
already reached a higher instance of repeated set agreement. 
Also, a value stored by a process in a lower instance is treated as
$\bot$.
Thus, a process decides in instance $t$ only if all tuples found in $A$ are stored by
processes in instance $t$ and there are at most $m$ distinct tuples, or
if another process has reached an instance higher than $t$.   

Each process $p$ maintains a local variable  $\textit{history}$ that
stores a sequence of output values that have been produced in the first
instances of repeated $k$-set agreement. 
In the current instance $t$, $p$ essentially follows the one-shot
algorithm (Figure~\ref{two-copies}), except that it appends the
current instance number $t$ and $\textit{history}$ to each value it
stores in the shared memory. 
Thus, each element of the vector returned by a scan of $A$
contains either $\bot$ or a tuple of the form
$(id,v,t',his)$.  
If $t'>t$, then $p_{id}$ has already completed instance
$t$ and  $his$ contains the corresponding output value.
If this is the case, $p$ adopts all the values output by $p_{id}$ for
instances from $t$ to $t'-1$.      
If $t'<t$, indicating that $p_{id}$ has not yet reached instance $t$, 
then the position of $A$ is treated as if it were $\bot$ in the
one-shot algorithm.

To prove {\bf $k$-agreement}, we focus on processes that produce
their output for instance $t$ without adopting a value from the history
that another process stored in $A$.  We call these \emph{$t$-deciding processes}.
Since each other processes that completes its $t$th {\sc Propose} adopts
one of the value of a $t$-deciding process, it suffices to prove that
$t$-deciding processes output at most $k$ different values.
%Again, we focus on $(n-\ell+1)$th such process $q_0$. Let $X$ denote the
%vector returned by its last snapshot and $V$ denote the set of values
%in $X$ ($|V|\leq |X|\leq m$). 
%Now we apply the arguments above to show that all $t$-deciding
%processes that output after $q_0$ must output values in $X$, i.e.,
%$n-\ell+m=k-m+m=k$ values are decided without adopting the outputs
%from faster processes. 
As in the proof for the one-shot case, we show that the last $\ell=n-k+m$ 
$t$-deciding processes output at most $m$ values.
There is one complication in the argument:
after the $(n-\ell+1)$th $t$-deciding process performs its last scan during
instance $t$, processes may store 
a $t'$-tuple with $t'<t$.
We show that each process can do this only in a single location, 
which ensures 
the agreement property for instance $t$ is not disrupted.
    
To show {\bf $m$-obstruction-freedom}, 
consider an execution where the set $P$ of processes that take infinitely many
steps has size at most $m$.
To derive a contradiction, assume some process in $P$ does not complete a {\sc Propose}.  Let $t$ be the smallest number for which some process does not complete its $t$th {\sc Propose} and let $P'$ be the set of processes that do not complete
their $t$th {\sc Propose}.  Since the processes in $P'$ never witness 
the presence of a process in a higher instance of set-agreement,
the argument for the one-shot case can be applied to this set $P'$
to obtain the desired contradiction.

A detailed proof of the algorithm can be found in Appendix~\ref{app:non-anon-algorithm}.
%the companion technical report~\cite{DFKR15-tr}. 

\begin{theorem}
\label{thm:repeated-alg}
For $1\leq m\leq k <n$, there is an $m$-obstruction-free algorithm that solve repeated  $k$-set agreement 
among $n$ processes using $\min(n+2m-k,n)$ registers.
\end{theorem}

% !TEX root =  podc-tr.tex

\section{Lower Bound for Anonymous
One-Shot Agreement}
\label{m-conc-lower-sec}
\indent

We now turn to anonymous algorithms, where processes are not equipped with identifiers and are programmed identically.
We also assume that the domain of possible input values is \nat.
In this section, we show that any $n$-process anonymous algorithm for $m$-obstruction-free (one-shot) $k$-set agreement requires $\Omega(\sqrt{\frac{nm}{k}})$ registers.  
Note that this bound on space complexity reflects all three parameters:  increasing $n$ or $m$ makes the problem harder
and increasing $k$ makes the problem easier.
It also generalizes the anonymous result of Fich, Herlihy and Shavit \cite{FHS98} (which is the special case when $m=k=1$) by showing  the
dependence on two additional parameters $m$ and $k$.
The assumption of anonymity allows us to add {\it clones} to an execution.  A clone of a process $p$ is another process $p'$ that has the same input as $p$.  Whenever $p$ takes a step, $p'$ takes an identical step immediately afterwards.

%Although the result is a generalisation of \cite{FHS98}, the technique used to prove it is 
%actually quite different:  when $k=1$, 
%the argument uses a kind of valency  argument, a technique that is
%not applicable to set agreement.
%The executions we construct must be stricter in only allowing each process $p$ to be aware of only those processes with the same input value as $p$.

%To make the problem solvable, we must have $m \leq k$.
%To make the problem non-trivial, we must have $k\leq n$.

\journalversion{For the journal version, we should see just how big the domain of possible values has to be to make the lower bound work.  It should work for exponential size domains.
See paper:  The minimum number of disjoint pairs in set systems and related problems by Das, Gan and Sudakov (which cites Erdos-Ko-Rado theorem of 1961 that might be useful).
In fact, the lower bound might even hold if the input domain has size n.
}

Let $A$ be an anonymous algorithm that solves $m$-obstruction-free $k$-set agreement among $n$ processes using finitely many registers.
For each set $V$  of $m$ distinct input values, fix an execution $\alpha(V)$ such that at most $m$ processes take steps during $\alpha$ and output
all values in~$V$.  (Such an execution exists, by Lemma \ref{lem:m-val}.)
Let $\bR(V)$  be the sequence of distinct registers written during
$\alpha(V)$ in the order they are first written in $\alpha(V)$.
For any sequence $\bR$ of distinct registers, define 
$\gV(\bR) = \{V \subset \nat : |V|=m \mbox{ and }\bR \mbox{ is a prefix of } \bR(V)\}$.

\begin{lemma}
\label{anonymous-gluing}
Let $r>0$ and
suppose $n\geq \ceil{\frac{k+1}{m}}(m+\frac{r^2-r}{2})$.  Then,
for $i=0,\ldots,r+1$, there is a sequence $\bR_i$ of length $i$ such that 
$\gV(\bR_i)$ is an infinite set.
\end{lemma}

\begin{proof}
We prove the claim by induction on $i$.  

Base case ($i=0$):  $\bR_0$ is the empty sequence and  $\gV(\bR_0) = \{V\subset \nat : |V|=m\}$ is infinite.

Induction step:  Let $i\in \{1,2,\ldots,r+1\}$.  Assume there is a sequence $\bR_{i-1}=\langle R_1, R_2, \ldots, R_{i-1}\rangle$ such that $\gV(\bR_{i-1})$ is infinite.

The induction step is  technical, so we begin with an informal overview.
Let $c=\ceil{\frac{k+1}{m}}$.
We first show that there cannot be $c$ disjoint sets $V_1,\ldots,V_c$ in $\gV(\bR_{i-1})$ such that each $\alpha(V_\ell)$ writes only to registers in $\bR_{i-1}$; otherwise, we could glue together the $\alpha(V_\ell)$'s %of those $\frac{k+1}{m}$ sets $V$ 
so that each $\alpha(V_\ell)$ is invisible to all the others, and 
the number of  output values in this glued-together execution would be
$|V_1\cup V_2\cup \cdots \cup V_c| = mc >k$.
Then, the rest of the argument is easy:  infinitely many sets in
$\gV(\bR_{i-1})$ must have register sequences of length at least $i$.
Since there are only finitely many  registers, infinitely many of
those sets  have the same register $R$ in position $i$ of their
sequence.  These form the infinite set $\gV(\bR_i)$, 
where $\bR_i = \bR_{i-1}\cdot R$.

%Let $c=\ceil{\frac{k+1}{m}}$.
To derive a contradiction, assume that  
(*) there exist $c$ disjoint sets $V_1,\ldots,V_c$ in $\gV(\bR_{i-1})$ such that
for all $\ell$, $\alpha(V_\ell)$ writes only to registers in $\bR_{i-1}$.
Let $P_1,\ldots,P_c$ be $c$ disjoint sets of $m$ processes each.
The following claim describes how we can glue together the $\alpha(V_\ell)$'s.
If $\beta$ is an execution  and $P$ is a set of processes, $\beta|P$ denotes the subsequence of $\beta$ consisting of steps taken by processes in $P$.

\medskip

{\sc Claim:}  For $j=0,1,\ldots,i-1$, there exists an execution $\beta_j$ with the following properties.
\begin{enumerate}
\item
\label{process-bound}
Exactly $\frac{cj(j-1)}{2}$ processes outside of $P_1\cup \ldots\cup P_c$ take steps during $\beta_j$.
\item
\label{write-exists}
For $\ell = 1,\ldots,c$, there is a write by some process in $P_\ell$ to each of $R_1,R_2,\ldots,R_j$ during $\beta_j$.
\item
\label{writes-contained}
No process writes to any register outside of $\{R_1,R_2,\ldots,R_j\}$ during $\beta_j$.
\item
\label{indistinguishable}
For $\ell= 1,\ldots,c$, $\beta_j | P_\ell$ is the prefix of $\alpha(V_\ell)$ up to but not including the first write to $R_{j+1}$ (or the entire execution $\alpha(V_\ell)$ if $j=i-1$).
\end{enumerate}

We prove the claim by inductively constructing the executions $\beta_j$.

{\sc Base case} ($j=0$):  We build $\beta_0$ by concatenating the maximal prefixes of $\alpha(V_1), \alpha(V_2), \ldots, \alpha(V_c)$ that do not contain any writes, performed by process sets $P_1, \ldots, P_c$, respectively.
No processes outside $P_1\cup \cdots\cup P_c$ take steps in $\beta_0$.
Property \ref{write-exists} is vacuously satisfied.  Properties \ref{writes-contained} and \ref{indistinguishable} follow 
immediately from the definition of $\beta_0$.

{\sc Inductive step}:  Let $j\in \{1,\ldots,i-1\}$.  Assume that there is a $\beta_{j-1}$ satisfying the four properties.
We describe how to construct $\beta_j$.

For each $\ell$, we insert $j-1$ clones of processes in $P_\ell$,
and we pause one clone just before the last write by a process in $P_\ell$ to each of $R_1,\ldots,R_{j-1}$.  Such a write exists, by property \ref{write-exists} of the induction hypothesis.
Moreover, there are enough processes to create these clones, since the number of processes that take steps in $\beta_{j-1}$ plus the $c(j-1)$ additional clones needed to construct $\beta_j$ total at most 
$mc+\frac{c(j-1)(j-2)}{2} + c(j-1) = mc + \frac{cj(j-1)}{2} \leq mc + \frac{c(i-1)(i-2)}{2} \leq mc+\frac{cr(r-1)}{2} = \ceil{\frac{k+1}{m}}(m+\frac{r^2-r}{2})$ and by the hypothesis of the lemma, there are this many processes in the system.

Let $\beta_{j-1}'$ be the execution that results from adding all of the clones to $\beta_{j-1}$.
We add some more steps to the end of $\beta_{j-1}'$ as follows.
For each $\ell = 1,\ldots, c$, we add a block write by the clones of processes in $P_\ell$
followed by steps of processes in $P_\ell$ continuing the steps of $\alpha(V_\ell)$ until
some process is poised to write to $R_{j+1}$ for the first time (or until the end of $\alpha(V_\ell)$ if $j=i-1$).  (This is legal, because the block write ensures that all
registers have the same state as they would have after $\beta_{j-1}|P_\ell$, which is a prefix of $\alpha(V_\ell)$, by induction hypothesis \ref{indistinguishable}.)  Thus, we ensure that $\beta_j$ satisfies property~\ref{indistinguishable}.

By property \ref{indistinguishable} of the inductive hypothesis, %it also follows that
the first newly added step by a process in $P_\ell$ writes to $R_j$.  Combined with induction hypothesis \ref{write-exists}, this proves property \ref{write-exists}.
For $j<i-1$, property \ref{writes-contained} holds because we stop the processes in
$P_\ell$ just before they write to any register outside of $\{R_1, \ldots, R_j\}$.
For $j=i-1$, property \ref{writes-contained} follows from our assumption (*) that $\alpha(V_\ell)$ writes only to registers in $\bR_{i-1}$.

The processes outside $P_1\cup \cdots \cup P_c$ that take steps in $\beta_j$
are the $\frac{c(j-1)(j-2)}{2}$ processes that take steps in $\beta_{j-1}$
plus the $c(j-1)$ clones that we added when constructing $\beta_{j-1}'$.
So the total number of such processes is $\frac{cj(j-1)}{2}$, satisfying property \ref{process-bound}.
This completes the proof of the claim.

\medskip

In $\beta_{i-1}$ processes in $P_\ell$ output all $m$ values in $V_\ell$ (for all $\ell$).  Since $V_1,\ldots,V_c$ are disjoint sets, there are at least $cm = \ceil{\frac{k+1}{m}}\cdot m \geq k+1$ different output values in $\beta_{i-1}$.
This contradicts the $k$-agreement property.
%fact that the algorithm solves $k$-set agreement.
Thus, assumption (*) is false, so
there are fewer than $c$ disjoint sets in $\gV(\bR_{i-1})$ such that $\alpha(V_\ell)$ writes only to registers in $\bR_{i-1}$.
Thus, there are infinitely many sets $V$ in $\gV(\bR_{i-1})$ such that $\alpha(V)$ writes outside of $\bR_{i-1}$.
Since there are only finitely many registers, there must be infinitely many of these sets $V$ such that the first register
outside of $\bR_{i-1}$ written during $\alpha(V)$ is the same for all $V$.  Call that register $R$.  Let $\bR_{i}$
be obtained by concatenating $R$ to the end of $\bR_{i-1}$.  Then, there are infinitely many sets $V$ such that $\bR_i$ is a prefix of
$\bR(V)$.
This completes the proof.
\end{proof}

\begin{theorem}
\label{m-conc-lower-thm}
Any anonymous algorithm that solves $m$-obstruction-free $k$-set agreement among $n$ processes using registers must use more than $\sqrt{m(\frac{n}{k}-2)}$ registers.
\end{theorem}

\begin{proof}
%Assume there is an algorithm $A$ that uses $r$ registers where $r\leq \sqrt{m(\frac{n}{k}-2)}$.
%By some algebraic manipulation (see Appendix \ref{anon-lower-proof}) we get $n\geq \ceil{\frac{k+1}{m}}\left(m+\frac{r^2-r}{2}\right)$.
%So, by Lemma \ref{anonymous-gluing}, there is a sequence of $r+1$ registers used in some execution of $A$, 
%which is impossible since $A$ only uses $r$ registers.
Assume an algorithm solves the problem using $r$ registers where $r\leq \sqrt{m(\frac{n}{k}-2)}$.
Then,
\begin{eqnarray*}
r & \leq  \sqrt{m\left(\frac{n}{k}-2\right)}\\
& \leq  \sqrt{m\left(\frac{2n}{k+m}-2\right)} \\
&(\mbox{ since }m\leq k \Rightarrow 
\frac{2n}{k+m}\geq\frac{n}{k} )\\
& = \sqrt{\frac{2m(n-k-m)}{k+m}}\\
\Rightarrow &  r^2-r \leq  \frac{2m(n-k-m)}{k+m}\\
\Rightarrow & \frac{k+m}{m}\cdot\frac{r^2-r}{2} \leq n-k-m\\
\Rightarrow & n  \geq \frac{k+m}{m}\left(m+\frac{r^2-r}{2}\right)\\
&\geq  \ceil{\frac{k+1}{m}}\left(m+\frac{r^2-r}{2}\right).
\end{eqnarray*}
 
So, by Lemma \ref{anonymous-gluing} there exists a sequence of $r+1$ registers used in some executions of $A$, 
which is impossible since there are only $r$ registers.
\end{proof}

\ignore{
In the special case where $m=k-1$, the arithmetic in the proof of Theorem \ref{m-conc-lower-thm} can be done more carefully to get a tighter bound, as described in the following proposition, which will be used in Section \ref{resilient-lower}.

\begin{proposition}
Any anonymous algorithm that solves $(k-1)$-obstruction-free $k$-set agreement among $n$ processes using registers must use more than $\sqrt{2(n-k+1)}$ registers.
\end{proposition}
\begin{proof}
Let $m=k-1$.
Assume an algorithm solves the problem using $r$ registers where $r\leq \sqrt{2(n-k+1)}$.
Then,
\begin{eqnarray*}
r -\frac{1}{2} & \leq & \sqrt{2(n-k+1)+\frac{1}{4}}\\
\left(r-\frac{1}{2}\right)^2 & \leq & 2(n-m) + \frac{1}{4}\\
r^2-r & \leq & 2(n-m) \\
n & \geq & m + \frac{r^2-r}{2}\\
n & \geq & \ceil{\frac{k+1}{m}}\left(m+\frac{r^2-r}{2}\right).
\end{eqnarray*}
So, by Lemma \ref{anonymous-gluing} there exists a sequence of $r+1$ registers used in some executions of $A$, 
which is impossible since there are only $r$ registers.
\end{proof}
}

\section{Anonymous Algorithm for
Repeated Set Agreement}
\label{m-conc-alg-sec}
\begin{theorem} 
\label{thm:anon-alg}
There is an algorithm that solves
$m$-obstruction-free repeated $k$-set agreement among $n$ processes (for $m\leq k$)
using $(m+1)(n-k)+m^2+1$ registers.
\end{theorem}

The anonymous algorithm presented in Figure~\ref{majority} solves
$m$-obstruction-free repeated $k$-set agreement among $n$ processes (for $m\leq k$)
using $(m+1)(n-k)+m^2+1$ registers.
%companion technical report~\cite{DFKR15-tr}.
The algorithm uses the same basic idea as the one  in Section \ref{linear-space-alg}.
It uses a snapshot object with $r=(m+1)(n-k)+m^2$ components,
which can be built anonymously and non-blocking using $r$ registers \cite{GR07}.
Again, the idea is to allow the first $\ell = n+m-k$ processes
to choose arbitrary outputs and then ensure that the last $n-\ell=k-m$ processes
output at most $m$ different values, for a total of at most $k$ different values.

\begin{figure*}
\setcounter{linenum}{0}
\begin{code}
\firstline
Shared variables:\nl
\n  $A$: snapshot object with $r= (m+1)(n-k)+m^2$ components, each initially $\bot$\nl  
    $H$: register, initially the empty sequence\bl
\p\nl
Persistent local variables:\nl
\n $i \leftarrow 0$\nl
   $t \leftarrow 0$\nl
   $\textit{history}\leftarrow$ empty sequence\bl
   
\nl
\p
{\sc Propose}($v$)\nl
\n  write $history$ into $H$\nl
    $t\leftarrow t+1$\nl
    if $|\textit{history}| \geq t$ then\nl
\n          output the $t$th value in \textit{history} and halt\nl
\p
   
    run the following two threads in parallel until one of them produces an output\nl

    thread 1:\nl
\n      {\it pref}  $\leftarrow v$\nl
        $\ell \leftarrow n+m-k$ \nl
        loop\nl
\n          update $i$th component of $A$ with value $(\mbox{\it pref},t,\textit{history})$\nl
		    $s\leftarrow $ scan  of $A$ \label{snap-anon}\nl
		    if $\exists j$ such that $s[j]=(w,t',his)$ with $t'> t$ then \nl
\n               $\textit{history} \leftarrow his$ \label{line:anon-hist1}\nl
                 output the $t$-th value in $his$ and halt \label{line:anon-exit1}\nl 
\p          if $|\{s[j] : 0\leq j < r\}| \leq m$ and every entry of
$s$ is a $t$-tuple then \nl
\n              let $w$ be the most common frequent value in $s$ \nl
                $\textit{history} \leftarrow \textit{history}\cdot w$ \label{line:anon-hist2}\nl
                output $w$ and
                halt \label{output-anon} \label{line:anon-exit2} \nl
\p		    if $|\{j : s[j]=(\mbox{\it pref},t,*)\}| < \ell$ and $\exists new$ such that $|\{j:s[j]=(new,t,*)\}|\geq \ell$ then \nl
\n                {\it pref} $\leftarrow new$\label{change-pref-anon}\nl
\p		    $i\leftarrow (i+1)\mymod r$ \label{change-index-anon}\nl
\p     end loop\nl
\p  thread 2:\nl
\n       loop \label{line:anon-t2begin}\nl
%\n           if $|H|\geq t$ then output $t$th element of $H$ and halt\nl
           \n if $|H|\geq t$ then  \nl
           \n              let $w$ be the $t$th element of $H$\nl
                $\textit{history} \leftarrow \textit{history}\cdot w$ \label{line:anon-hist3}\nl
                 output $w$ and halt \label{output-anon-2} \label{line:anon-exit3}\nl
\p\p       end loop \label{line:anon-t2end}\nl
\p
\p end {\sc Propose}
\end{code}
\caption{Anonymous algorithm for $m$-obstruction-free repeated $k$-set agreement.\label{majority}}
\end{figure*}

For one-shot $k$-set agreement, processes
alternate between storing their preferred value in a component of the
snapshot object $A$
and performing a scan of $A$.
The conditions for outputting a value and adopting a new preference differ
from the algorithm in Section \ref{linear-space-alg} to compensate for the lack
of identifiers.
Whenever a process observes $m$ or fewer different values in a scan, it can 
output the one that occurs most frequently.
Otherwise, if a process sees fewer than $\ell$ copies of its own preference and
at least $\ell$ copies of another value, it adopts this other value as its preference.

The adaptation of this algorithm to repeated consensus is similar to the
technique used for the non-anonymous case.  There is one additional
complication:  there is no known space-efficient wait-free anonymous snapshot
implementation from registers, so we use a non-blocking implementation.
Therefore, some processes may \emph{starve} while accessing the snapshot
object, under the condition that at least one process manages to
complete infinitely many instances of $k$-set agreement.  

To ensure that starving processes
also complete their {\sc Propose} operations we use one additional
register $H$ where ``fast'' processes write their outputs.
%Each starving process can learn the outcome of the instance of agreement it is working on.
Every process periodically checks $H$ in a parallel thread
(lines~\lref{line:anon-t2begin}-\lref{line:anon-t2begin}) 
and if it finds out that $|H|\geq t$, where $t$ is the 
instance of agreement the process is working on, it outputs the
$t$-th value found in $H$.  
As in the non-anonymous case, the sequence of values that have been
output in the instances of repeated $k$-set agreement the process has
completed so far is stored in a local variable  $\textit{history}$.
To ensure that $\textit{history}$ is updated exactly once per
instance of $k$-set agreement, we require that
the threads of a process are scheduled so that the pairs of
lines~\lref{line:anon-hist1}--\lref{line:anon-exit1}, 
\lref{line:anon-hist2}--\lref{line:anon-exit2},
and \lref{line:anon-hist3}--\lref{line:anon-exit3}
are executed without interruption from the process's other thread.

The proof of correctness of our algorithm is given in  Appendix~\ref{anonymous-alg}.

\section{Concluding Remarks}
\indent 

\journalversion{Could mention that lower bounds apply even for the
weaker termination condition of $m$-obstacle-freedom (also defined by Taubenfeld)
which says that if at most m processes take infinitely many steps, then
some process completes its operation (but starvation of individual processes
is allowed).}

A small gap remains between the upper and lower bounds for non-anonymous repeated
set agreement.
The one-shot algorithm of \cite{DFGR13} uses fewer registers than ours for one special case:  when $m=1$ and $k=n-1$, it uses two registers compared to our three.  This suggests the upper bound could perhaps be improved to $n+m-k$.
The gaps are larger for the other scenarios shown in Figure \ref{summary}.
It would be interesting to see if there is an anonymous algorithm that uses
linear space, rather than quadratic space.
Another natural continuation of this work would be to extend the one-shot anonymous
lower bound to the non-anonymous setting.
However, closing the gap for the one-shot setting eludes us still.

%\bibliographystyle{plain}
%\bibliography{podc15}

\newpage

% !TEX root =  podc15.tex
\newpage
\appendix

\section{Proof of correctness of repeated set agreement}
\label{app:non-anon-algorithm}

In this section, we prove Theorem~\ref{thm:repeated-alg}.
The pseudocode for our repeated $k$-set agreement algorithm appears in
Figure~\ref{mconc-rep-kset}.  
It essentially follows the pseudocode of the  one-shot algorithm
(Figure~\ref{two-copies}), with additional ``shortcuts'' which a
process may use to adopt a value output previously by another process that has
already reached a higher instance of repeated set agreement. 
Also, a value stored by a process in a lower instance is treated as
$\bot$.
Thus, a process decides in instance $t$ only if all tuples found in $A$ are stored by
processes in instance $t$ and there are at most $m$ distinct tuples, or
if another process has reached an instance higher than $t$.     
The local variable {\it history} initially stores an empty sequence and the local
variable $t$ is initially 0.  
The local variable {\it i} stores the location that the process updates and is initially 0. 
The values of these three local variables persist from one
invocation of {\sc Propose} to the next.
In particular, this means that the first location  of a 
{\sc Propose} is the last location of the previous {\sc Propose}. 

A process updates components of the shared snapshot object with tuples of the form $(v,id,t,\textit{history})$, where $v$ is the process's preferred value, $id$ is the identifier of the process, 
$t$ indicates which instance of set agreement the process is currently working on,
and \textit{history} is a sequence of output values for instances of set agreement.
We refer to a tuple whose third element is $t$ as a $t$-tuple.

To see that the algorithm satisfies {\bf validity}, first observe
that when a process invokes {\sc Propose} for the $t$th time, the length of its
\textit{history} variable is at least $t-1$.
The value in every $t$-tuple in $A$ and, thus, put in the $t$th position of a 
process's local variable $\textit{history}$, is the input value of some process's $t$th invocation
of {\sc Propose}.

The following Lemma reformulates Lemma~\ref{invOne} for
$t$-tuples, showing that $A$ cannot contain more than one distinct $t$-tuple
for a given process.

\begin{lemma}\label{invOne:rep}
Let $id$ be a process identifier and $t$ be a positive integer. In any reachable configuration,
all $t$-tuples with identifier $id$ in $A$ are identical. 
\end{lemma}

\begin{proof}
To derive a contradiction, assume that in some reachable configuration $C$,
$A[i_1]=(v_1,id,t,his_1)$ and $A[i_2]=(v_2,id,t,his_1)$ such that
$(v_1,his_1)\neq (v_2,his_1)$.
Let $p_{id}$ be the process with identifier $id$.
By the algorithm, $p_{id}$ changes its $\textit{history}$ variable only when it
switches to a higher instance of repeated agreement.  
Thus, $his_1=his_2$ and we must have $v_1\neq v_2$.
Let $C$ be reached in some execution at time $\mu$.
Let $u_1$ and $u_2$ be the last update steps before $\mu$ in which $p_{id}$ updates 
$A[i_1]$ and $A[i_2]$, respectively.
Without loss of generality, assume that $u_1$ occurred before $u_2$.
Then, at some time between $u_1$ and $u_2$, $p_{id}$ changes its
{\it pref} variable in instance $t$ (at line~\lref{change-pref-rep}).  
Consider the first time after $u_1$ when $p_{id}$ performs such a change, and let
$i^*$ and $s^*$ be the values of $p_{id}$'s local variables $i$ and $s$ at
that time. Since (1) $A[i_1]=(v_1,id,t,his_1)$ at all times between $u_1$ and $\mu$ and
(2) $s^*$ is obtained between $u_1$ and $\mu$, $s^*[i_1]$ must be equal
to $(v_1,id,t,his_1)$. By the algorithm, $i^*=i_1$; otherwise, the test in line~\lref{cond-rep} would not be satisfied,
and $p_{id}$ would not change {\it pref} in line~\lref{change-pref-rep}.
Therefore, in the next iteration of the loop, $p_{id}$ will update location $A[i_1]$.
This update is after $u_1$ and no later than $u_2$ (and hence before
$\mu$), which contradicts the definition of $u_1$ as the last update
performed by $p_{id}$ to $A[i_1]$ before~$\mu$.
\end{proof}

To show {\bf $k$-agreement}, we use arguments similar to the proof for the one shot algorithm.
Let $\ell=n-k+m$.
We call a process \emph{$t$-deciding} if it outputs a value at line \lref{output-rep} (i.e., without
adopting a value from another process's \textit{history} value) during
its $t$th invocation of {\sc Propose}.
If, for a given instance $t$, 
at most $n-\ell$ processes are $t$-deciding, 
then $k$-agreement for instance $t$ is immediate since $n-\ell=k-m<k$.   
Otherwise, consider an execution in which more than $n-\ell$ processes
are $t$-deciding.
Order these processes according to the time that they perform their last scan in instance $t$,
and let $q_0$ be the $(n-\ell+1)$th process in this ordering.
Let $X$ be the set of at most $m$ different tuples that appear in
$q_0$'s final scan and $V$ be the set of values in $X$.  Then, $|V|\leq
|X|\leq m$.
We shall show that $q_0$ and all processes that come later in the ordering
output values in $V$.  
%\modif{or that have been decided by the $n-\ell$ previous processes}.
Thus, the total number of output values in instance $t$ 
is at most $(n-\ell) + |V| \leq n-(n-k+m)+m = k$.

\begin{lemma}\label{lem:safety:rep}
After $q_0$ performs its final scan in instance $t$, 
only $t$-tuples with values in $V$ can appear twice in $A$.
\end{lemma}
\begin{proof}
This proof is analogous to the proof of Lemma \ref{lem:safety} for the one-shot algorithm.
Let $C_0$ be the configuration just after $q_0$'s last scan.
We shall show by induction that each configuration
reachable from $C_0$, only $t$-tuples with values in $V$ can appear in two or more locations of $A$.
For the base case, consider the configuration $C_0$.
%$A$ only contains pairs in $X$
%and, since $|X|=m<n+2m-k$, 
By the definition of $V$, $A$ contains only tuples with values in $V$,  
so the claim holds.

For the induction step, suppose the claim holds in all configurations from $C_0$ to some configuration $C_1$ reachable from $C_0$.  Let $st$ be a step that takes the system from $C_1$ to another configuration $C_2$.  We must show that the claim holds in configuration $C_2$.
We need only consider steps $st$ in which some process $p_{id}$ stores a tuple $(v,id,t,his)$ in $A$.

{\sc Case 1}: $st$ is the first time $p_{id}$ stores a $t$-tuple after $C_0$. 
If $v\in V$, then $st$ cannot cause a violation of the claim.
If $v\notin V$, then $A$ contains exactly one copy of $(v,id,t,his)$ in configuration
$C_2$, so again $st$ preserves the claim.

{\sc Case 2}: $st$ is not the first time $p_{id}$ stores a $t$-tuple after $C_0$.
Let $s_{id}$ be the vector obtained by $p_{id}$'s
last scan (at line \lref{snap-rep}) before $st$.
Since $s_{id}$ is not in the last iteration of the loop during instance $t$, $s_{id}$ must not
satisfy the conditions on line \lref{halt-cond1-rep} or \lref{halt-cond2-rep}.
We show that $v\in V$, and hence $st$ preserves the claim, by considering two subcases.

{\sc Case 2a}:  $s_{id}$ satisfies the condition on line \lref{cond-rep}. 
Since the condition on line \lref{halt-cond1-rep} is not satisfied and the condition on line
\lref{cond-rep} is satisfied, every tuple in $s_{id}$ is a $t$-tuple.
Then, $p_{id}$ updates its {\it pref} variable at line \lref{change-pref},
so the value $v$ is the value
of a $t$-tuple that appears twice in $s_{id}$.
By the induction hypothesis, $v\in V$.

{\sc Case 2b}: $s_{id}$ does not satisfy the condition on line \lref{cond-rep}.

We call an update after $C_0$ {\it bad} if it stores either a $t'$-tuple with $t'<t$
or a $t$-tuple that is not in $X$.
We first argue that each process can do bad updates to at most one location.
To derive a contradiction, suppose some process does bad updates to two different
locations after $C_0$.
Consider the first process $p$ to do a bad update to a second location.
%\MP{modified the proof here}
%Between $p$'s last bad update to one location and its first bad update to the second 
%location, $p$ must
%execute line \lref{change-ind-rep}.
Process $p$'s last bad update to one location and its first bad update to the second 
location must be in the same instance of {\sc Propose}, because $p$ must
execute line \lref{change-ind-rep} between them.
Let $s_p$ be the vector returned by the scan that $p$ performs at line \lref{snap-rep} during the
iteration of the loop when it executes line \lref{change-ind-rep}.
Then, $s_p$ must not satisfy the conditions on line \lref{halt-cond1-rep} or \lref{cond-rep}.
Recall that at least $n-\ell+1$ processes have updated $A$ for the last time during instance $t$
prior to $C_0$.
So at most $\ell-1$ processes can do bad updates.
Since no process has done bad updates to two locations before the $p$'s scan obtained the vector $s_p$,
and no location of $s_p$ contains a tuple with instance number greater than $t$,
at least $r-\ell+1 = m+1$ locations of $s_p$ contain $t$-tuples in $X$.
Since $|X|\leq m$, at least two locations of $s_p$ contain the same $t$-tuple.
This contradicts the fact that $s_p$ does not satisfy the condition on line \lref{cond-rep}.
Thus, each process can do bad updates to at most one location.

Hence, at all times after $C_0$, at least $r-(\ell-1) = m+1$ locations have not
had any bad updates performed on them.
Since $s_{id}$ did not satisfy the condition on line \lref{halt-cond1-rep},
$s_{id}$ must contain at least $m+1$ $t$-tuples in $X$, and therefore $s_{id}$ contains
at least two identical $t$-tuples.
Moreover, some process $q_0$ satisfied the condition on line \lref{halt-cond2-rep} prior
to the scan that returned $s_{id}$, so no component of $s_{id}$ contains $\bot$.
Thus, the only reason $s_{id}$ does not satisfy the condition on line \lref{cond-rep}
must be that for some $j$ different from $p_{id}$'s position $i$,
$s_{id}[j]=(v,id,t,his)$.  Just before taking the scan $s_{id}$, $p_{id}$ updates
location $i$ with $(v,id,t,his)$.  This update occurs after $C_0$, since $st$ is 
not the first update by $p_{id}$ after $C_0$.  
In the configuration after this update to location $i$,
both $s_{id}[j]$ and $s_{id}[i]$ contain $(v,id,t,his)$.
So, by the induction hypothesis, $v\in V$.
\end{proof}
Lemma~\ref{lem:safety:rep} implies that all $t$-deciding processes after the $(n-\ell)$th  
output values in $V$ and, thus, a total of at most $n-\ell+m=k$ values are output by 
$t$-deciding processes.
The {\bf $k$-agreement} property follows. 

To prove {\bf $m$-obstruction-freedom}, consider an execution where the set $P$ of 
processes that take infinitely many steps has size at most $m$.
To derive a contradiction, assume that some process in $P$ 
completes only a finite number of {\sc Propose} operations. 
Let $t$ be the smallest number such that a process in $P$ does not complete its $t$th {\sc Propose}.
Let $P'$ be the set of processes in $P$ that do not complete the $t$th {\sc Propose}.
By the algorithm, no process in $P'$ ever witnesses the presence of a
process in a higher instance; otherwise, it would output a value
decided in instance $t$ at line \lref{halt-rep}.
  
Eventually, processes stop storing tuples with instance numbers
$t'<t$ in $A$.  Below we reuse the arguments of the proof of
Lemma~\ref{term-dontchange} to show that at least one process in $P'$
updates each component of $A$ infinitely often.

Recall that each time a process in $P'$ executes the loop in
instance $t$, it either keeps its preferred value and increments $i$ (the next location to
update) modulo $r$ or changes its preferred value without modifying  $i$. 
Let  $NS$ denote the set of processes in $P'$ that  increment $i$ infinitely often and
the set $S$ denotes the rest of the processes in $P'$, 
i.e., the processes that eventually get stuck
updating to the same location forever.
% and forever changing their preferred
%values each time they execute the loop.

\begin{lemma}\label{term-dontchange:rep}
$NS\neq \emptyset$.
\end{lemma}
\begin{proof}
The proof is by contradiction.
Assume it is not the case ($P'=S$).

Let $M$ be the set of at most $m$ locations that processes in $S$ eventually settle on.
Note that no process in $P-P'$ can update a location outside of $M$ infinitely often
because then the processes in $P'$ would eventually see a tuple with instance number greater
than $t$ and complete their $t$th {\sc Propose} operation.
Let $\mu$ be a time after which
only processes in $P$ take steps and 
no process updates a location outside of $M$.
Let $NM$ be the set of at least $n+m-k\ge 2$ locations that are never changed after $\mu$.

Since all positions in $NM$ that contain tuples of earlier instances
are ignored, we simply reuse the arguments of the proof of
Lemma~\ref{term-dontchange}, to derive a contradiction.
\end{proof}
By Lemma~\ref{term-dontchange:rep}, (1)~there is a time after which only tuples stored
by processes in $P'$ are found in scans performed by processes in $P'$, 
and all of them are $t$-tuples.
By~Lemma~\ref{invOne:rep}, (2)~all $t$-tuples in $A$ of the same process are identical
and (3)~$|P'|\leq|P|\leq m$. (1), (2) and (3) imply that there is a time after which,  whenever a process $p\in P'$ performs a scan, 
it finds at most $m$ different $t$-tuples in the returned vector and,
thus, decides, contradicting the definition of $P'$.  This completes the proof of the {\bf $m$-obstruction-freedom} property. 

Thus, we have shown that the algorithm solves repeated $k$-set agreement using
a snapshot object with $n+2m-k$ registers, which can be implemented using
$\min(n,n+2m-k)$ registers, as described in the proof of Theorem \ref{one-shot-alg}.
This completes the proof of Theorem \ref{thm:repeated-alg}.

% !TEX root = podc15.tex

\section{Proof of correctness of anonymous repeated set agreement}
\label{anonymous-alg}
\indent

To prove Theorem~\ref{thm:anon-alg}, consider our algorithm
in Figure~\ref{majority}.
The algorithm actually uses a non-blocking snapshot object with $r=(m+1)(n-k)+m^2$ components,
which can be built anonymously using $r$ registers \cite{GR07}, plus one additional register.
Each component of the snapshot object is initially $\bot$.

In this algorithm, a process stores tuples of the form $(v,t,history)$ where 
$v$ is the process's preferred value, $t$ indicates which instance of set
agreement the process is currently working on, and $history$ is a sequence
of output values for instances of set agreement.  We refer to a tuple whose
second element is $t$ as a $t$-tuple.

As an invariant, it is easy to see that each of the following can only store input values
of some process's $t$ invocation of {\sc Propose}:
\begin{itemize}
\item
a process's {\it pref} variable during the process's $t$th invocation of {\sc Propose},
\item
the first component of a $t$-tuple appearing in $A$, and
\item
the $t$th element of any sequence that is stored in a process's {\it history} variable, in the shared variable $H$ or inside a tuple in $A$.
\end{itemize}
{\bf Validity} follows.

Next, we prove the {\bf $k$-agreement} property.
A process is \emph{$t$-deciding} if it outputs a value on line \lref{output-anon}.
Any other process that produces an output for its $t$th {\sc Propose} operation
outputs the same result as some $t$-deciding process, so it suffices to show
that the $t$-deciding processes output at most $k$ different values.
As in Section \ref{linear-space-alg}, we show that the last $\ell=n-k+m$ $t$-deciding processes output at most $m$ different values, so that the total number of outputs for instance $t$ is at most $n-\ell+m=k$ values.

If at most $n-\ell$ processes are $t$-deciding, then $k$-agreement is trivial for the $t$th instance of set agreement, since
$n-\ell = k-m < k$.
So, consider an execution in which more than $n-\ell$ processes are $t$-deciding.
Order the $t$-deciding processes according to the time that they perform their last scan in their $t$th invocations of {\sc Propose},
and let $q_0$ be the $(n-\ell+1)$th process in this ordering.
Let $X$ be the set of tuples that appear in $q_0$'s final scan.
Let $V$ be the values that appear in tuples in $X$.
We prove that $q_0$ and all $t$-deciding processes that come later in the ordering
output values in $V$.

We call an update of $A$ after $C_0$ a {\it bad update} if it stores
a $t'$-tuple with $t'<t$ or a $t$-tuple whose value is not in $V$.

\begin{lemma}
After $q_0$ performs its final scan in its $t$th {\sc Propose} operation, each process performs bad updates to at most one 
component of $A$.
\end{lemma}
\begin{proof}
To derive a contradiction, assume that some process performs bad updates to two components of $A$ 
after $q_0$'s final scan $scan_0$.
Consider the first process $p$ to do a bad update on a second location.
Let $s_p$ be the vector returned by the last scan that $p$ performs before its bad update to the second location.
This scan causes $p$ to execute line \lref{change-index-anon}
so that it can perform an update on the second location.
Thus $s_p$ does not contain any $t'$-tuple with $t'>t$.
Since $n-\ell+1$ processes have performed their final scan of their $t$th {\sc Propose}
operation at or before $scan_0$, at most $\ell-1$ processes can perform
updates that store $t'$-tuples with $t'\leq t$ after $scan_0$.
By definition of $p$, none of those $\ell-1$ processes 
have performed bad updates on two different locations between
$scan_0$ and  $p$'s scan that returned $s_p$.
Since $scan_0$ returned a vector that contained only $t$-tuples, $s_p$ must contain at most $\ell-1$ components
that are either $t'$-tuples with $t'<t$ or $t$-tuples with values not in $V$.
So there are at least $r-\ell+1 = (m+1)(\ell-1)+1-(\ell-1)=m(\ell-1)+1$ locations of $s_p$ that contain $t$-tuples with values in $V$.
Since $|V|\leq m$, one of the values in $V$ must appear in $t$-tuples stored in at least $\ell$ locations.
Thus $p$ must adopt a value in $V$ after it obtains the scan $s_p$, contradicting
the fact that $p$'s next update after this scan uses a value not in $V$.
\end{proof}

It follows that at any time after $q_0$'s final scan, there are at most $\ell-1$
$t$-tuples in $A$ with values that are not in $V$.
Any $t$-deciding process ordered after $q_0$ performs a final scan
that returns only $t$-tuples, so one of the values in $V$ must appear in at least
$\ell$ of them, and is therefore the most frequent value in the scan.
Thus, the value output by any such process must be in $V$.
Hence, the total number of  values output is at most $(n-\ell) + |V| \leq n-(n-k+m)+m = k$, ensuring that $k$-agreement is satisfied.

Finally, we prove {\bf $m$-obstruction-freedom}.
For this part of the proof, it is convenient to consider lines \lref{snap-anon}
to \lref{change-index-anon} to be a single atomic action.
Since there is only one shared-memory access in this block of code,
there is no loss of generality in this assumption:  every execution
has an equivalent execution where this block is executed atomically,
so if we prove $m$-obstruction-freedom for executions that satisfy this assumption 
then it also holds for all executions.

Consider an execution where at most $m$ processes continue to take steps forever.
Let $P$ be the set of processes that complete infinitely many accesses
to the snapshot object.  $P$ is non-empty, since the snapshot implementation
we use is non-blocking, and $|P|\leq m$.
To derive a contradiction, assume
that some process in $P$ never completes one of its {\sc Propose} operations.
Let $t$ be the smallest number such that some process in $P$ does not complete
its $t$th {\sc Propose}.
Let $P'$ be the set of processes in $P$ that do not complete their $t$th {\sc Propose}
operation.
Let $\mu$ be a time after 
\begin{itemize}
\item
every process outside $P$ has stopped performing updates on $A$,
\item
every process in $P'$ has begun its $t$th {\sc Propose} operation, 
\item
every process in $P-P'$ has begun its $(t+1)$th {\sc Propose} operation, and
\item
no component of $A$ contains a $t'$-tuple for any $t'<t$.
\end{itemize}
It is possible to choose $\mu$ to satisfy the last condition because
each process in $P'$ completes infinitely many iterations of the loop
and therefore updates
every location of $A$ after $\mu$.
Thus, eventually
all $t'$-tuples with $t'<t$ are overwritten.
Note that after $\mu$, no component of $A$ ever contains a $t'$-tuple with $t'<t$.

We say that a value $v$ is a {\it candidate} in a configuration $C$ if it is either the {\it pref} value of some process in $P'$
or it appears in a $t$-tuple in $A$.
We shall prove that there is a configuration after $\mu$ with at most $m$ candidates.  
After that point, only those $m$ values can appear in $t$-tuples in the snapshot object.
It follows that every process in $P'$ completes its $t$th {\sc Propose}
when it next performs a scan, which contradicts the definition of $P'$.

\begin{lemma}
\label{disappear}
If, in some configuration $C$ after $\mu$, a value $v$ is not the 
{\it pref} of any process in $P'$ and $t$-tuples with value $v$ appear in fewer than 
$\ell$ components of the snapshot object, then after some 
time, $v$ will not be a candidate anymore.
\end{lemma}
\begin{proof}
To derive a contradiction, assume that some process in $P'$ changes its
local {\it pref} variable to $v$ in some step after $C$.
Consider the first such step by any process after $C$.
Let $scan$ be the scan performed in that step.
Between $C$ and $scan$, no process executing its $t$th {\sc Propose} stores
a $t$-tuple with value $v$, so the result of $scan$ contains $t$-tuples with value $v$
in at most $\ell-1$ components, contradicting the fact that $p$ adopts the value $v$
in the step when it performs $scan$.

Thus, no process in $P'$ ever has $v$ as its preferred value after $C$.
So, no $t$-tuple with value $v$ is ever stored in $A$ after $C$.
Since each process in $P'$ executes infinitely many steps of its $t$th
{\sc Propose} operation, and increments its index $i$ in every iteration of
the loop, it eventually overwrites every component of $A$.
Thus, there is a time (after $C$), after which no component of $A$ contains a
$t$-tuple with value $v$.
After this time, $v$ is never a candidate.
\end{proof}

\begin{lemma}
\label{choice-available}
Whenever a process in $P'$ performs a scan after $\mu$, 
there is some value $v$ that appears in $t$-tuples in at least 
$\ell$ components of $A$.
\end{lemma}
\begin{proof}
To derive a contradiction, suppose there is no such value $v$.  
Consider the configuration $C$ immediately after the scan.
By Lemma \ref{disappear}, only the values stored in {\it pref} variables of processes
in $P'$ remain candidates forever.  There are at most $m$ such values.
Thus, there is a time after which every $t$-tuple in $A$ contains only those values.
Whenever a process in $P'$ performs a scan after that time, it will
terminate, contradicting our assumption that no process in $P'$ ever completes
its $t$th {\sc Propose}.
\end{proof}

%
%For convenience, we consider a shared-memory access and all of the process's local computation after it to be a single, atomic action.  Thus, each process alternates between two atomic actions:  (1) taking a snapshot and either deciding or updating its preference and index, and (2) writing a value.
%We say that a process is {\it poised} if the action it will perform when it is next scheduled to take a step will be of type (2).

For any configuration $C$ and value $v$, let $mult(C,v)$ be the number of components of $A$ that contain $t$-tuples with value $v$ in $C$ plus the number of poised processes that are poised to perform an update and have {\it pref} $v$ in $C$.
The following lemma generalizes Lemma \ref{disappear}.

\begin{lemma}
\label{disappear2}
Consider a value $v$.  If, in some configuration $C$ after $\mu$, $mult(C,v) < \ell$, then after some time, $v$ will no longer be a candidate.
\end{lemma}
\begin{proof}
We first show that if a single step $st$ takes the system from a configuration $C_1$ to another configuration $C_2$ and $mult(C_1,v)<\ell$ then $mult(C_2,v)<\ell$.
If $st$ is a step by a process in $P-P'$, it can only decrease $mult$.
If $st$ is an update by a process in $P'$, 
$st$ may increase by one the number of components of $A$ containing a $t$-tuple with
value $v$, but then $st$ will also decrease the number of processes poised
to store a $t$-tuple with value $v$ by one, so the value of $mult$ cannot
be increased by $st$.
Finally, suppose $st$ is an atomic execution of lines \lref{snap-anon}--\lref{change-index-anon}.
In $C_1$, fewer than $\ell$ components of $A$ contain $t$-tuples with value
$v$ (since $mult(v,C_1)<\ell$).  Moreover, by Lemma \ref{choice-available}, there
is a value $v'$ such that $t$-tuples with value $v'$ appear in at least $\ell$
components of the scan performed during $st$.  Thus,
the process performing $st$ adopts some value different from $v$ as its {\it pref}.
So, $st$ cannot increase $mult$ for $v$.

Thus, in every configuration reachable from $C$, $mult(C,v)<\ell$.
As argued above, any process in $P'$ that performs a scan after $C$ adopts
a value different from $v$.
Thus, eventually, no process will have its {\it pref} equal to $v$, and at that time, $v$ will be in at most $\ell-1$ components of $A$, so Lemma \ref{disappear} ensures that $v$ will eventually cease to be a candidate.
\end{proof}

Now, consider a configuration $C$ immediately after some process has performed an
update (after $\mu$).
There are $(m+1)(\ell-1)+1$ registers and at most $m-1$ processes in $P'$ poised
to perform an update.  Thus, $\sum\limits_v mult(v,C) \leq (m+1)\ell-1$.  Therefore, at most $m$ values have $mult(C,v) \geq \ell$.  All other values will eventually cease to be candidates, by Lemma \ref{disappear2}, so eventually there will be at most $m$ candidates.  All processes in $P'$ will then terminate when they next perform a scan, which contradicts our definition of $P'$.

Thus, we have shown that every process in $P$ completes infinitely many {\sc Propose} operations.
There remains one more thing to show.  There may be some processes not in $P$
that takes infinitely many steps.  (These are processes that starve in the
non-blocking implementation of the snapshot object.)
We must show that each such process $p$ also completes all of its {\sc Propose} operations.
Processes in $P$ write longer and longer sequences to $H$ infinitely often and processes not in $P$ eventually stop writing to $H$.
Thus, for all $t$, $p$ will eventually see a sequence in $H$ of length at least $t$,
and will then complete its $t$th {\sc Propose} operation.

This completes the proof of Theorem \ref{thm:anon-alg}.
We remark that for the one-shot case, the register $H$ is not
required, so we can solve the one-shot version using one less register.

\end{document}